\documentclass{amsart}

\usepackage{amssymb,amsfonts,latexsym,amsmath,amsthm,wrapfig,graphicx, hyperref, 
}

\usepackage{bm}
\usepackage{epsfig}

\input xy
\xyoption{all}

\numberwithin{equation}{section}
\newtheorem{thm}{Theorem}[section]
\newtheorem{prop}[thm]{Proposition}
\newtheorem{problem}{Problem}

\newtheorem{remark}{Remark}
\newtheorem{definition}{Definition}
\newtheorem{conjecture}{Conjecture}

\newcommand{\RR}{\mathbb{R}}
\newcommand{\CC}{\mathbb{C}}
\newcommand{\ol}{\overline}

\newcommand{\supp}{\rm{supp}}

\newcommand{\p}{\partial}

\newcommand{\be}{\begin{equation}}
\newcommand{\ee}{\end{equation}}
\newcommand{\la}{\label}

\begin{document}

\title[Energy Equilibria for Point Charge Distributions]{Around a theorem of F. Dyson and A. Lenard:
Energy Equilibria for Point Charge Distributions in Classical Electrostatics}

\date{December 2019. \\  $\quad \quad ^\sharp$The author's research is supported by the Simons Foundation, under the grant 513381.}

\author[A. Abanov]{Artem Abanov}
\email{abanov@tamu.edu}
\address{MS 4242,Texas A$\&$M University, College Station, TX 77843-4242}

\author[N. Hayford]{Nathan Hayford}
\email{nhayford@mail.usf.edu}
\address{4202 E. Fowler Ave., CMC342, Tampa, FL 33620}

\author[D. Khavinson]{Dima Khavinson$^\sharp$}
\email{dkhavins@usf.edu}
\address{4202 E. Fowler Ave., CMC342, Tampa, FL 33620}

\author[R. Teodorescu]{Razvan Teodorescu}
\email{razvan@usf.edu}
\address{4202 E. Fowler Ave., CMC342, Tampa, FL 33620}

\begin{abstract}
We discuss several results in electrostatics: Onsager's inequality, an extension of Earnshaw's theorem,  and a result 
stemming from the celebrated conjecture of Maxwell on the number of points of electrostatic equilibrium. Whenever 
possible, we try to provide a brief historical context and references. 


\end{abstract}

\maketitle

\hfill {\it{Dedicated to the memory of Freeman Dyson (1923 -- 2020).\\}} 

\section{Introduction.}
Electrostatics is an ancient subject, as far as most mathematicians and physicists are concerned. After several 
centuries of meticulous study by people like Gauss, Faraday, and Maxwell (to name a few), one wonders if it is 
still possible to find surprising or new results in the field. Throughout this text, we address several seemingly
classical electrostatics problems that have not been fully addressed in the literature, to the best of our knowledge. 
Let us begin by establishing some notations.

For vectors $x,y \in \mathbb{R}^d, \, d \geq 3$, we define the function $K_{d}(x,y)$ by the formula
	\begin{equation} \label{Riesz}
		K_{d}(x, y) = \frac{1}{(2-d) \omega_{d-1}}\frac{1}{|x-y|^{d-2}}.
	\end{equation}
Here, $\omega_{d-1}$ is the surface area of the unit sphere in $\RR^d$. $K_{d}$ is the fundamental 
solution for the Laplace operator in $\mathbb{R}^d$ (i.e., $\Delta_y K_{d}(x, y) = \delta_x$). 

Furthermore, given a locally finite (signed) Borel measure $\mu$ with support $\Sigma$, 
we define the \textit{Newtonian (or Coulomb) potential} of $\mu$ with respect to the kernel $K_{d}(x,y)$ by
	\begin{equation}
		\left(U^{\mu}_{\Sigma}\right)(x) =
		\int_{\Sigma} K_{d}(x,y) d\mu(y).
	\end{equation}
If the support of the measure $\mu$ is either clear from context or, otherwise, irrelevant to the problem, 
we drop the subscript $\Sigma$, and write $U^{\mu}(x)$.

We define the \textit{Newtonian (or Coulomb) energy} of a measure $\mu$ with respect to the kernel $K_{d}(x,y)$ by
	\begin{equation}
		W_{\Sigma}[\mu] = \int_{\Sigma} \left (U^{\mu}_{\Sigma}\right)(x) d\mu(x).
	\end{equation}
We also refer to this functional as the \textit{electrostatic energy}, and sometimes use the simpler notation 
$W_{\Sigma}[\mu] = W[\mu]$ - cf. \cite{Kellogg, Landkof} for the basics of Potential Theory.

\begin{definition}
We say a charge distribution (measure) $\mu$ is in \textit{constrained equilibrium}, when its electrostatic
potential 
	\begin{equation} \la{1}
		U^{\mu}_{\Sigma}(x) = \int_{\Sigma} K_d(x,y) d\mu(y), 
	\end{equation}
is constant (possibly taking different values) on each connected component $\Sigma_j$ of the support of $\mu$, 
subject to the constraints 
	\begin{equation} \la{2}
		\mu_j = \mu(\Sigma_j) = Q_j, \, j = 1, 2, 3, \ldots, m, \quad  \sum_{j=1}^m Q_j = Q. 
	\end{equation}
\end{definition}

In the case when $\Sigma_j$ consists only of one point, i.e. $\Sigma_j = \{ x_j \}$, the condition that the 
potential $U^{\mu}_{\Sigma}$ be constant is replaced by the gradient condition 
	\begin{equation} \la{3}
		\left (  \nabla U^{\mu}_{\Sigma \setminus \Sigma_j}  \right ) (x_j)= 0.
	\end{equation}
 
Now, we are ready to state the first problem discussed in this paper.
\begin{problem} \la{p1}
Given a total charge $Q \in \mathbb{R}$ and a locally finite Borel measure $\mu$, such that 
$\mu(\mathbb{R}^n) = Q$, is there a collection of disjoint compact sets $\{ \Sigma_j \}_{j = 1}^m$, such that 
$\mu(\Sigma_j) = Q_j$, $ \sum_{j=1}^m Q_j = Q$,  and $\mu$ has a constrained equilibrium configuration on 
$\Sigma := \cup_j \Sigma_j$? 
\end{problem}	

It should be noted that, while the associated energy 
	\begin{equation} \la{4}
		W_{\Sigma}[\mu] \equiv \int_{\Sigma} U^{\mu}_{\Sigma}(x) d\mu(x) 
	\end{equation}
does solve a variational problem over the set $\mathcal{M}$ of measures constrained by \eqref{2}, it is by no means 
automatically also the solution to the free optimization problem 
	\begin{equation} \la{5}
		W[\mu]  = \inf_{\sigma \in \mathcal{M}} W[\sigma].
	\end{equation}

In other words, a solution for Problem~\eqref{p1} merely gives an equilibrium charge configuration, which need not 
be also a \textit{stable equilibrium}, i.e. a local minimum of the energy functional, as opposed to a saddle point. 
The (stronger) stable equilibrium problem may not have a solution, for a generic choice of the support $\Sigma$ and 
set of constraints \eqref{2}.

\begin{remark} \normalfont
It should also be noted that the choice of total charge $Q$ is not important, except to distinguish between neutral 
configurations ($Q = 0$) and non-neutral ($Q \ne 0$). This is due to the fact that, under a simple rescaling 
\begin{equation*}
	Q = \lambda q, \, \, Q_j = \lambda q_j, \,\, \lambda \in \mathbb{R}\setminus \{0\},
\end{equation*}
the potential scales by a factor of $\lambda$, and the energy by a factor of $\lambda^2$, leaving the variational 
problems (and their solutions) unchanged. Therefore, the only two distinct cases that need to be considered are $Q 
= 0$ and $Q = 1$. 
\end{remark}	

The second (classical) problem we discuss is the characterization of critical manifolds $C$ 
(specifically, curves) on which the gradient of the potential of a given charge configuration vanishes. More 
precisely, the problem in $\RR^2$ is:

\begin{problem} \la{p2}
Let $\mu$ be a locally finite charge distribution with planar support  $\Sigma \subset \RR^2$. Can the critical
manifold $C \subset \RR^3$ of $\mu$, defined by 
	\begin{equation} \la{8}
	C = \{ x \in \mathbb{R}^3 : \nabla U^{\mu}_{\Sigma}(x) = 0 \}
	\end{equation}
contain a curve in $\RR^2$? 
\end{problem}	

This problem has numerous applications, some of which are discussed in this paper. One of its most obvious 
implications is that a collection of charges, placed in the plane supporting the distribution $\mu$, cannot have 
a curve in the same plane as an equilibrium configuration. 
	
Problem \ref{p2} has a distinguished history, and can be regarded as a special case of Maxwell's Problem 
(cf. \cite{Max} for further references). In short, Maxwell asserted (without proof) that if $n$ point charges 
$\{x_j\}_{j=1}^n$ are placed in $\RR^3$, then there are at most $(n-1)^2$ points in $\RR^3$ at which the 
electrostatic field vanishes. More precisely, if each $x_j$ has charge $q_j \in \RR$, then 
	\begin{equation}\label{Maxwell-bound}
		\# \bigg\{ x : \nabla \bigg( \sum_{j=1}^n \frac{q_j}{|x-x_j|}\bigg) = 0 \bigg\} \leq (n-1)^2, 
	\end{equation}
or is, otherwise, infinite (e.g. contains a curve, the `degenerate case'). Thompson, while preparing
Maxwell's works for publication, couldn't prove it, and the problem has since become known as Maxwell's conjecture -
cf. \cite{Max} for more details. So far, the only real progress on this problem has been achieved in \cite{Max},  
where it was verified that the cardinality of the set of isolated points in equilibrium in \eqref{Maxwell-bound} 
is finite. But even in the case of $n = 3$, the best known estimate is $12$, not $4$! No counterexamples to the 
conjecture have been found.

The first examples of curves of degeneracy where \eqref{Maxwell-bound} holds go back to \cite{Y}. Futher, partial 
results, for a particular case when the charges are coplanar can be found in \cite{Killian,Peretz}. Killian, in 
particular, conjectured that in the latter case the degeneracy curves are all transversal to the plane of the
charges. This is proven in \S4. Finally, note that in the plane, if we use the logarithmic potential, the estimate
improves to $(n-1)$  and is an obvious corollary of the Fundamental Theorem of Algebra.
	
This paper is organized as follows: in \S2, we discuss Problem~\ref{p1}, first in its classical form (for the
Newtonian potential in $\mathbb{R}^3$), and the proof of F. Dyson and A. Lenard for an inequality first discovered by 
L. Onsager \cite{Onsager}, along with extensions of the same energy inequality to the case of potentials 
in $\mathbb{R}^n$. 

In \S3 we focus on necessary conditions for the existence of an equilibrium configuration, in particular for the 
case of Coulomb potentials (in $\mathbb{R}^3$). A necessary condition independent of the support, and which can be 
expressed as a constraint on the measure density moments, is also discussed in \S3 (Intersection Theorem). 
	
Section \S4 is dedicated to the precise formulation and solution of degeneracy in Maxwell's problem, 
for charge configurations constrained to two-dimensional subspaces of $\RR^3$. In section \S5 we pose a fascinating
question, originating in approximation theory, that we frivolously label `Faraday's problem', believing that
Sir Michael Faraday would have never hesitated to answer it based on empirical evidence. 
	
\section{Variations on Onsager's Inequality.}
	 The Onsager inequality was originally discussed by Lars Onsager \footnote{Lars Onsager was a theoretical 
	 physicist and chemist. He was best known for his work in statistical mechanics, in particular his 
	 eponymous relations, which won him the Nobel Prize in Chemistry in 1968, and for his exact solution of the 
	 2D Ising model. For further details about the life and work of Onsager, see \cite{LO}.} in a relatively 
	 little known paper \cite{Onsager}. 
	 Onsager himself did not provide a proof of this inequality, and it was not until 30 years
	 later that a full proof was given by Freeman Dyson and Andrew Lenard \cite{Dyson-Lenard}. 
	 Here, we shall present their original proof of the inequality, and consider subsequent generalizations 
	 to $\mathbb{R}^n$. We remark that this inequality was brought to the attention of the authors by 
	 Eero Saksman et. al., who found far-reaching extensions of it in a probablistic context \cite{Saksman}.
	\subsection{F. Dyson and A. Lenard's Original Proof.}
		Let $\{x_j\}_{j=1}^n$ be a collection of point charges in $\RR^3$, with charges $\{q_j\}_{j=1}^n$,
		each of which takes one of the values $\pm 1$. The electrostatic energy due to this collection of point 
		charges is given by
			\begin{equation}
				-\frac{1}{4\pi} \sum_{j < k} \frac{q_j q_k}{|x_j - x_k|}.
			\end{equation}
		Furthermore, for each $x_j$, denote the shortest distance to the next point charge by $\delta_j$:
			\begin{equation}
				\delta_j := \min_{k \neq j} |x_j - x_k|.
			\end{equation}
		The Onsager inequality states that the electrostatic energy of the point charges is bounded by the sum of 
		inverses of $\delta_j$'s:
			\begin{prop} \label{original}
			Let $\{x_j\}_{j=1}^n$, $\{q_j\}_{j=1}^n$, and $\{\delta_j\}_{j=1}^n$ be defined as above. Then
				\begin{equation} \label{Onsager}
					\frac{1}{4\pi} \sum_{j < k} -\frac{q_j q_k}{|x_j - x_k|} <
					\frac{1}{4\pi} \sum_{j=1}^n \frac{1}{\delta_j}.
				\end{equation}

			\end{prop}
		We remark that the inequality agrees with the intuition that spreading the charges out over a larger
		distance will decrease their electrostatic energy.
		\begin{proof}
		We now replicate Dyson and Lenard's original proof of the inequality. It is based on the intuition 
		that replacing point charges with uniform distributions on spheres of the same total charge will not 
		change the electrostatic energy of the configuration, along with the fact that the total energy is 
		always positive. Consider a distribution of charges $\mu$ supported on spheres of radii $\{\rho_j\}_{j=1}^n$
		centered at $\{x_j\}_{j=1}^n$, each carrying the total charge $\{q_j\}_{j=1}^n$. In other words, the 
		sphere $B_j$ centered at $x_j$ with radius $\rho_j$ carries uniform surface charge density 
		$\frac{d\mu_j}{dA} = \frac{q_j}{4\pi \rho_j^2}$ (here, $dA$ denotes the surface area measure). 
		At any point $x$ outside of these spheres, this distribution of charge generates the potential
			\begin{equation*}
				U(x) = \sum_{j=1}^n \frac{q_j}{4\pi |x - x_j|}.
			\end{equation*}
		The desired inequality will follow from the positivity of the total energy $W$ of this distribution 
		of charge:
			\begin{equation*}
				W = \int U(x) d \mu (x) = \int_{\RR^3} |\nabla U(x)|^2 dx \geq 0.
			\end{equation*}
		The latter calculation is well known in potential theory (see \cite{Landkof}, for example), and is a 
		straighforward corollary of Green's formula. We can decompose the total energy into two parts, the self-energy 
		of the spheres and the mutual pairwise energy: 			
		\begin{equation*}
				W = \sum_{j=1}^n \int_{B_j}U_{\{q_j\}}d\mu_j + 
				\sum_{k\neq j, k = 1}^n \sum_{j=1}^n q_j U_{\{q_k\}}(x_j)  = 
				\sum_{j=1}^n\frac{q_j^2}{4\pi \rho_j} + 2 \sum_{1 \leq j < k \leq n}
				\frac{q_j q_k}{4\pi|x_j - x_k|}.
			\end{equation*}
		Therefore, the positivity of the total energy can be rewritten as
			\begin{equation} \label{inequality1}
				\sum_{j=1}^n \frac{q_j^2}{4\pi \rho_j} > 2 \sum_{1 \leq j < k \leq n}
				 -\frac{q_j q_k}{4\pi|x_j - x_k|}.
			\end{equation}
		Notice that the right hand side is just the expression for the electrostatic energy of the $n$ point
		charges. We have some freedom in the above inequality in picking the radii of the spheres; one particularly
		good choice will be to pick the radii to be as large as possible without having any of the spheres
		intersect; thus, we choose
			\begin{equation}
				\rho_j = \frac{1}{2}\min_{k \neq j}|x_j - x_k| = \frac{1}{2}\delta_j
			\end{equation}
		Then the inequality \eqref{inequality1} becomes
			\begin{equation}
				\sum_{j=1}^n \frac{q_j^2}{\delta_j} > \sum_{1 \leq j < k \leq n}
				 -\frac{q_j q_k}{|x_j - x_k|}.
			\end{equation}
		Finally, in the case where  $q_j \in \{\pm 1\}$, we obtain the original Onsager's inequality:
			\begin{equation}
				\sum_{j=1}^n \frac{1}{\delta_j} > \sum_{1 \leq j < k \leq n}
				 -\frac{q_j q_k}{|x_j - x_k|}.
			\end{equation}
		\end{proof}
		\begin{remark} \normalfont
		The reader may have noticed that the inequality relies on redistribution of the point charges
		over spheres, and the object of interest is really the self-energy of the spheres (i.e. the expression on
		the left hand side of equation \eqref{inequality1}). One may wonder if this inequality can be improved by 
		redistributing the total charge of each point charge differently, for example by replacing the point charge
		with a uniform volume distribution over the ball of radius $\rho$ (of equal total charge). In fact, the
		uniform surface distribution yields the optimal estimate from this perspective, since this distribution of
		charge yields the smallest self-energy. This follows readily from the fact that the equilibrium measure of
		the ball in $\RR^3$ is the uniform distribution on its surface.
		\end{remark}
	\subsection{Generalizations of the Onsager Inequality.}
		We can easily generalize the Onsager inequality to $\RR^d$, by considering the appropriate electrostatic 
		(or, Coulomb) potential in $\RR^d$:
		\begin{prop} \label{Or1}		
		 For $d > 2$, and with all notations adopted from Section 1, the Onsager inequality becomes
			\begin{equation} \label{q2}
				2^{d-3} \sum_{j=1}^n \frac{q_j^2}{\delta_j^{d-2}} > 
				-\sum_{1\le j<k \le n} \frac{q_j q_k}{|x_j - x_k |^{d-2}}.
			\end{equation}
		\end{prop}
		The proof is virtually identical to the one in $\RR^3$; we provide a sketch here. 
		Consider a distribution of charges $\mu$ supported on spheres of radii $\{\rho_j\}_{j=1}^n$ centered 
		at $\{x_j\}_{j=1}^n$, each carrying total charge $\{q_j\}_{j=1}^n$. In other words, the sphere
		centered at $x_j$ with radius $\rho_j$ carries uniform surface charge density $\frac{q_j}{\omega_{d-1}
		\rho_j^{d-1}}$. (Here, as in \S1, $\omega_{d-1}$ denotes the surface area of the unit sphere in $\RR^d$.) 
		Again, the pivotal fact is the positivity of the total energy: 
			\begin{equation}
				W = \sum_{j=1}^n \frac{q_j^2}{\omega_{d-1} \rho_j^{d-2}} + 
				2\sum_{1\leq j<k \leq n}\frac{q_j q_k}{\omega_{d-1} |x_j - x_k |^{d-2}} > 0.
			\end{equation}
		As before, a sharper upper bound is provided by letting each sphere become tangent to its 
		nearest neighbor, i.e., by choosing:
			\begin{equation}
				\rho_j = \frac{1}{2} \min_{k \ne j} |x_j - x_k| = \frac{1}{2} \delta_j.
			\end{equation}
		Thus, we obtain that
			\begin{equation}
				2^{d-3} \sum_{j=1}^n \frac{q_j^2}{\delta_j^{d-2}} > 
				\sum_{1\le j<k \le n} -\frac{q_j q_k}{|x_j - x_k |^{d-2}}. 
			\end{equation}
		Finally, letting $q_j = \pm 1$ yields an inequality equivalent to \eqref{Onsager}:
			\begin{equation} \label{ineq}
				2^{d-3} \sum_{j=1}^n \frac{1}{\delta_j^{d-2}} >
				\sum_{1\le j<k \le n} -\frac{q_j q_k}{|x_j - x_k |^{d-2}}. 
			\end{equation}
		\begin{remark} \normalfont
			We remark that \eqref{ineq} can hold for any number of point charges; both sides of the inequality 
			can be made arbitrarily small  by choosing a configuration of charges with the nearest distance
			$\delta = \min_j{\delta_j}$ to be sufficiently large. Thus, Onsager's inequality is strict, and only
			becomes an equality when the charges are moved away to infinity.
		\end{remark}
	In $d = 2$ dimensions, the Coulomb interaction is no longer a power law; moreover, the total energy of a 
	distribution of charge is no longer necessarily positive. The positivity of the total energy is indispensible in 
	the proof above; thus no version of the Onsager inequality as general as the one for $\RR^d$ ($d > 2$) exists. 
	However, if we impose additional conditions 
	\footnote{For example, one could consider only charge configurations with total charge $\sum_j q_j = 0$; such
	configurations are guaranteed to have positive energy. Alternatively, if one imposes that the charges are all
	confined to the unit disc, positivity of total energy is again ensured. The proof of these facts can be found in
	\cite{Landkof}.} 
	to guarantee the energy postivity constraint, an Onsager-like inequality can be written down. However, the result
	is rather artificial, since it is only valid for very specific charge configurations; thus, we omit it.
\section{Intersection Theorem.}
			As we have already seen, two-dimensional electrostatics is rather special; the following theorem is no
		exception to this rule. The theorem may have been known earlier, and the authors would be interested to 
		see if one could find the earliest instance of it in the literature. We remark that it is a theorem about
		point charges in \textit{unstable} equilibrium, in comparison to the celebrated Earnshaw's theorem, which  
		ensures that unstable equilibria are the only nontrivial equilibrium configurations. In other words, 
		the potential of an electrostatic field cannot have local minima (or maxima) in space free of charges, 
		only saddle points (the proof, of course, is obvious, since the electrostatic potential is a harmonic 
		function away from the support of the charges --  cf. \cite{Earnshaw}).
		
		\subsection{Intersection Theorem: A Necessary Condition for Equilibrium.}
			Let $\{z_i\}_{i=1}^n$ be point charges in the complex plane, with corresponding charges 
			$\{q_i\}_{i=1}^n$. Assume that the charges are in \textit{equilibrium}, i.e. the electrostatic force
			acting on each charge is zero:
				\begin{equation} \label{equilibrium}
					\sum_{j \neq i} \frac{q_j}{z_i - z_j} = 0,
				\end{equation}
			for each $i = 1, ... , n$ (The field intensity is $-\frac{1}{\pi z}$, as a consequence of the fact that
			the Coulomb kernel is $\frac{1}{2\pi}\log\frac{1}{|z|}$). The Intersection Theorem provides 
			necessary conditions for these charges to be in equilibrium:
			\begin{prop} \label{prop3.1}
			Let $\{z_i\}_{i=1}^n$ be point charges (with corresponding charges $\{q_i\}_{i=1}^n$) in equilibrium
			in the complex plane. Then, necessarily,  
				\begin{equation} \label{Abanov}
					\sum_{i=1}^n q_i^2 = \bigg( \sum_{i=1}^n q_i \bigg)^2.
				\end{equation}
			Furthermore, for each $k \geq 0$, we must also have that
				\begin{equation} \label{eq-relations}
					(k+1)\sum_{i=1}^n q_i^2 z_i^k = \sum_{\ell = 0}^k 
					\bigg(\sum_{i=1}^n q_i z_i^\ell \bigg)\bigg(\sum_{j=1}^n q_j z_j^{k-\ell}\bigg).
				\end{equation}
			\end{prop}
			
			\begin{remark} \label{sphere} \normalfont
			The name chosen for this result follows from a geometric interpretation of Eq.~\eqref{Abanov}: 
			consider the point in $\mathbb{R}^n$, with coordinates $\{q_k\}_{k=1}^n$. Then a simple way of 
			interpreting \eqref{Abanov} is to say that it describes the intersection between the hyperplanes 
			$\sum_k q_k = \pm |Q|$ and the sphere of radius $|Q|$, centered at the origin. For example, this shows 
			that for $n = 2$ the only solution is trivial, i.e. only one charge can be non-zero. 
			\end{remark}
			
			\begin{remark} \normalfont
				For $n > 1$, \eqref{eq-relations} implies that the charges $\{q_j\}_{j=1}^n$, if we think of them
				as vectors in $\CC^n$ with real coordinates, must satisfy infinitely many quadratic equations (i.e.,
				lie in the intersection of infinitely many quadrics \eqref{eq-relations} in $\CC^n$, with 
				coefficients depending on the positions $z_j \in \CC$ where the charges sit. The special case 
				\eqref{Abanov} implies that, if $\sum q_j = Q = 0$, equilibrium never occurs.
				
				Obviously, the configurations that are in equilibrium (and hence satisfy the infinitely many 
				equations  \eqref{eq-relations}) are very special. But they do exist! For example, if we 
				equidistribute $n-1$ equal charges $q$ at the vertices of a regular $(n-1)$-gon on the unit 
				circle $\{|z| = 1\}$, and then place a charge $q_n := -\frac{q (n-2)}{2}$ at the center of the
				circle, the total force acting on each charge will be zero. Hence, for this configuration,
				\eqref{eq-relations} (and also \eqref{Abanov}) hold.		
			\end{remark}			
			
			\begin{proof}
			Define functions $g(z) = \prod_{i=1}^n (z-z_i)^{q_i}$, and $G(z) = (\partial_z \log g(z) )^2$ 
			($G(z)$ can be interpreted as the square of the complex electric field). Consider the expansion of 
			$G(z)$ near $\infty$ in two different ways: first, write 
				\begin{align}
					G(z) &= \sum_{i = 1}^n \frac{q_i^2}{(z-z_i)^2} + 
					\sum_{i \neq j} \frac{q_i q_j}{z_i - z_j}\bigg( \frac{1}{z-z_i} - \frac{1}{z-z_j}\bigg) 
					\label{E-squared}\\
					&= \sum_{i = 1}^n \frac{q_i^2}{(z-z_i)^2} \nonumber,
				\end{align}
			using the equilibrium condition \eqref{equilibrium} by first summing over $i$ to get rid of the second 
			term in \eqref{E-squared}. Expanding $G(z)$ at infinity, we find that
				\begin{equation*}
					G(z) = \sum_{k=0}^\infty \frac{1}{z^{k+2}} \bigg((k+1)\sum_{i=1}^n q_i^2 z_i^k\bigg).
				\end{equation*}
			Now, expand $G(z)$ without taking into account the condition for equilibrium:
				\begin{equation*}
					G(z) = \sum_{k=0}^\infty \frac{1}{z^{k+2}} 
					\bigg\{\sum_{\ell = 0}^k \bigg(\sum_{i=1}^n q_i z_i^\ell\bigg) 
					\bigg(\sum_{j=1}^n q_j z_j^{k-\ell}\bigg)\bigg\}.
				\end{equation*}
			Since these expansions must be the same, we can equate their coefficients termwise and obtain
				\begin{equation}
					(k+1)\sum_{i=1}^n q_i^2 z_i^k = \sum_{\ell = 0}^k 
					\bigg(\sum_{i=1}^n q_i z_i^\ell \bigg)\bigg(\sum_{j=1}^n q_j z_j^{k-\ell}\bigg),
				\end{equation}
			for each $k \geq 0$. 
			\end{proof}			
			Every one of the above conditions must necessarily hold for the charges to be in equilibrium. 
		\subsection{Generalization to Compactly Supported Charge Distributions.}
			The Intersection Theorem may be generalized to compactly supported charge distributions, as well.
			\begin{prop}\label{prop3.2}
			Let $\rho(z)$ be a continuous density of charge compactly supported on a bounded domain 
			$\Omega \subset \CC$. Suppose again the charges are in equilibrium. Then, necessarily,
				\begin{equation}
					(k+1)\int_\Omega \rho^2(\zeta) \zeta^k dA(\zeta) = 
					\sum_{\ell = 0}^k \bigg(\int_\Omega \rho(\zeta)\zeta^\ell dA(\zeta) 
					\int_\Omega \rho(\zeta)\zeta^{k -\ell} dA(\zeta)\bigg), \, \forall k \in \mathbb{N}, 
				\end{equation}
			where $dA$ denotes the Lebesgue area measure. 
			\end{prop}
			\begin{proof}		
			 The analog of the function $G(z)$ becomes:
				\begin{equation}
					\widetilde{G}(z) = \int_\Omega \int_\Omega \frac{\rho(\zeta) \rho(w)}{(z-\zeta)(z-w)} 
					dA(\zeta) dA(w).
				\end{equation}
			Again, let us compute $\widetilde{G}(z)$ in two different ways: first, we rewrite $\widetilde{G}(z)$ as
				\begin{equation*}
					\widetilde{G}(z) = 
					\int \int_{\zeta \neq w} \frac{\rho(\zeta) \rho(w)}{(z-\zeta)(z-w)} dA(\zeta) dA(w) 
					+\int_\Omega \frac{\rho^2(\zeta)}{(z-\zeta)^2} dA(\zeta).
				\end{equation*}
			The second integral is understood in the Cauchy Principal Value sense and is known as the
			Hilbert transform, or Beurling transform  \footnote{The Beurling transform
			is the most studied example of the class of Calder\'{o}n-Zygmund operators. In short, define
				\begin{equation*}
					T\rho(z) := \int_\Omega \frac{\rho^2(\zeta)}{(z-\zeta)^2} dA(\zeta) = \lim_{\epsilon\to 0}
					\int_{\Omega\cap\{|z-\zeta| > \epsilon\}} \frac{\rho^2(\zeta)}{(z-\zeta)^2} dA(\zeta).
				\end{equation*}
			Then $T$ is a bounded operator from $L^2$ to $L^2$ (with respect to the area measure) and preserves
			smooth functions.} in 2D -- cf. \cite{Ahlfors, Dra}.
			
			The first integral can be rewritten as
				\begin{align*}
					\int \int_{\zeta \neq w} \frac{\rho(\zeta) \rho(w)}{(z-\zeta)(z-w)} dA(\zeta) dA(w) 
					&= \int \int_{\zeta \neq w} \frac{\rho(\zeta) \rho(w)}{\zeta-w} 
					\bigg(\frac{1}{z-\zeta} - \frac{1}{z-w}\bigg) dA(\zeta) dA(w) \\
					&= \int_\Omega \frac{\rho(\zeta)}{z-\zeta}
					\bigg[\int_{w\neq \zeta} \frac{\rho(w)}{\zeta - w}dA(w)\bigg]dA(\zeta) \\
					&-
					\int_\Omega \frac{\rho(w)}{z-w}
					\bigg[\int_{\zeta\neq w} \frac{\rho(\zeta)}{\zeta - w}dA(\zeta)\bigg]dA(w) \\
					&= 0,
				\end{align*}
			as $\int_{\zeta\neq w} \frac{\rho(\zeta)}{\zeta - w}dA(\zeta) = 0$ is the condition for equilibrium.
			Therefore, we have that	
			\begin{equation*}
				\widetilde{G}(z) = \int_\Omega \frac{\rho^2(\zeta)}{(z-\zeta)^2} dA(\zeta).
			\end{equation*}					
			Expanding this expression, we find that
				\begin{equation*}
					\widetilde{G}(z) = \sum_{k=0}^{\infty} \frac{1}{z^{k+2}} \bigg((k+1)\int_\Omega \rho^2(\zeta) 
					\zeta^k dA(\zeta)\bigg).
				\end{equation*}
			On the other hand, expanding $\widetilde{G}(z)$ without taking into account the condition for 
			equilibrium, we obtain 
				\begin{equation*}
					\widetilde{G}(z) = \sum_{k=0}^{\infty} \frac{1}{z^{k+2}} \sum_{\ell = 0}^k 
					\bigg(\int_\Omega \rho(\zeta)\zeta^\ell dA(\zeta) 
					\int_\Omega \rho(\zeta)\zeta^{k -\ell} dA(\zeta)\bigg). 
				\end{equation*}
			Equating the coefficients, we find the following sequence of conditions for equilibrium:
				\begin{equation}
					(k+1)\int_\Omega \rho^2(\zeta) \zeta^k dA(\zeta) = 
					\sum_{\ell = 0}^k \bigg(\int_\Omega \rho(\zeta)\zeta^\ell dA(\zeta) 
					\int_\Omega \rho(\zeta)\zeta^{k -\ell} dA(\zeta)\bigg).				
				\end{equation}					
			\end{proof}			
			In particular, for $k = 0$ we obtain the expected continuous analog of Intersection Theorem \eqref{Abanov}:
				\begin{equation}
					\int_\Omega \rho^2(\zeta)dA(\zeta) = 
					\bigg( \int_\Omega \rho(\zeta) dA(\zeta)\bigg)^2.
				\end{equation}
		\begin{remark} \normalfont
		It is interesting to note that the necessary condition for $k=0$ does not involve the actual 
		configuration $\{ z_j \}$, but only the values of the charges $\{ q_j\}$. This peculiar fact can be 
		traced back to the different scaling behavior of equilibrium configurations in $\mathbb{R}^2$ versus 
		$\mathbb{R}^d, \, d \ne 2$: upon scaling an equilibrium configuration $\{ x_j \} \to 
		\{ \lambda x_j \}, \lambda > 0$, the total energy scales as $\lambda^{1-d}$ for $d \ne 2$, but 
		for $d = 2$ it only acquires an additive term proportional to the difference between the two sides in 
		\eqref{Abanov}:
		
		\begin{equation*}
		W[\{ \lambda z_j \}] - W[\{z_j\}] = \frac{1}{2\pi} \left ( \left (\sum_j q_j \right )^2 - \sum_j q_j^2\right )
		\ln \lambda
		\end{equation*}
	
		Therefore, for $d = 2$, the necessary condition \eqref{Abanov} follows from the fact that, if it were not 
		satisfied, moving all the charges according to an infinitesimal dilation would lead to a lower energy 
		($\lambda < 1$ if the right-hand side in \eqref{Abanov} is larger than the left, and $\lambda > 1$ otherwise), 
		which would mean the initial configuration was not in equilibrium. For all $d \ne 2$, this reasoning fails, 
		and any necessary condition seems to require the explicit dependence on the configuration itself
		(positions $\{ x_j\}$). It would be interesting to pursue this theme further.
		
		Furthermore, Propositions \ref{prop3.1}, \ref{prop3.2} provide \textit{necessary} conditions for equilibrium;
		one wonders if, taken collectively, they are indeed also \textit{sufficient}. We think it is a compelling 
		and possibly challenging problem to determine either (a.) if these equations are sufficent for equilibrium, 
		or, if they are not, (b.) find a corresponding collection of sufficient conditions.
		\end{remark}	
		\begin{remark} \normalfont
		The above theorem admits an easy generalization to any number of dimensions. Consider a pairwise interaction
		between a collection of particles $\{x_j\}_{j=1}^n \subset \RR^d$. Assume that the energy of the interaction
		depends only on the charges of the particles $\{q_j\}$ and the pairwise distances between them $|x_i  - x_j|$;
		these are natural assumptions. The total pairwise energy of the $n$-particle configuration is then given by
			\begin{equation}
				W := \sum_{i,j, i \neq j} q_i q_j \Phi(|x_i - x_j|),
			\end{equation}
		where $\Phi$ is some given function characterizing the interaction. Then, the force acting on particle $i$ is
			\begin{equation} 
				F_i = -\nabla_{x_i} \bigg( \sum_{j \neq i} q_i q_j \Phi(|x_i - x_j|) \bigg),
			\end{equation}
		and the equilibrium condition then is that $F_i = 0$, $i = 1, ..., n$, i.e.
			\begin{equation}\label{cnd}
				\sum_{j \neq i} q_i q_j \frac{x_i - x_j}{|x_i - x_j|}\Phi'(|x_i - x_j|) = 0,
			\end{equation}
		$i = 1, ..., n$. We can rewrite \eqref{cnd} through the logarithmic derivative of $\Phi$, as 
		\begin{equation}
		\sum_{j \neq i} q_i q_j \frac{x_i - x_j}{|x_i - x_j|^2} \left [r \Phi'(r) \right ]_{r = |x_i - x_j|}
				= 0. 
		\end{equation}
		Denote 
		\begin{equation} \label{mass}
		M_{ij} :=  \frac{q_i q_j }{|x_i - x_j|^2} \left [r \Phi'(r) \right ]_{r = |x_i - x_j|},
		\end{equation}
		then obviously $M_{ij} = M_{ji}, i \ne j$. Setting $M_{jj} : = 0, j = 1, 2, \ldots, n$, \eqref{cnd} becomes
		\begin{equation*}
		\sum_{j } M_{ij} (x_i -x_j) = 0, \,\, \forall \, i = 1, 2, \ldots, n,
		\end{equation*}
		so we can multiply each equation by $x_i$ and sum over $i$, to obtain (since $M_{ij} = M_{ji}$) 
		\begin{equation*}
		\sum_{i,j } M_{ij} (x_i -x_j)\cdot x_i = 0 \Rightarrow \sum_{i,j } M_{ij} (x_j -x_i)\cdot x_j = 0
		\end{equation*}
		Adding these equations and using $M_{jj} = 0$ leads to the general form of \eqref{Abanov} 
		\begin{equation} \label{ct}
		\sum_{i \ne j } M_{ij} |x_i -x_j|^2 = 0 \Rightarrow 
		\sum_{i \neq j} q_i q_j \left [r \Phi'(r) \right ]_{r =|x_i - x_j|}
				= 0. 
		\end{equation}
		If $\Phi(r) = -\log(r)$, then $r\Phi' = -1$, and \eqref{ct} is equivalent to \eqref{Abanov}. If $
		\Phi(r) = r^{-k}$, $k>0$, then $r\Phi' = -k \Phi$, so \eqref{ct} becomes 
		\begin{equation}
			W = \sum_{i,j, i \neq j} q_i q_j \Phi(|x_i - x_j|) = 0,
		\end{equation}
		or, equivalently, if the equilibrium exists, the total energy of the system is zero. 
		\end{remark}
\section{Degeneracy in Maxwell's Problem with Planar Charge Distributions.}
	
	We now address Problem~\ref{p2}, which stems from Maxwell's conjecture. The following question was first 
	discussed  in \cite{Killian}, also cf. \cite{Peretz}: 
	\begin{prop}
	Consider a distribution of point charges $\mu$ with support contained in a plane $\mathcal{H} \simeq \mathbb{R}^2$
	(without loss of generality, we take  $\mathcal{H}$ to be the $xy$-plane). Then  the critical manifold 
	$C \subset \mathbb{R}{^3}$ of $\mu$, defined by 
		\begin{equation}
			C = \{ x \in \mathbb{R}^3 : \nabla U^{\mu}(x) = 0  \}
		\end{equation}
	cannot contain a curve in $\mathcal{H}$. 
	\end{prop}	
	In other words, if $C$ contains a curve on which $\nabla U^{\mu}(x) = 0$, then the latter is necessarily 
	transversal to the plane $\mathcal{H}$.
	\begin{proof}
		To see this, let us assume $C$ contains a curve in $\mathcal{H}$. By a slight abuse of notation, we shall still 
		denote it by $C$. Since the support of $\mu$ is in the plane, we have that 
		$\frac{\partial U^{\mu}}{\partial z} = 0$ in $\mathcal{H}$. Now, consider the analytic hypersurface 
		$\Gamma_\mu := \{(x,y)\in\CC^2 : U^{\mu}(x,y,0) = const.\}$ in $\CC^2$. On the curve $C = \Gamma_\mu \cap
		\mathcal{H}$, we have that 
		\begin{equation*}
			\bigg(\frac{\partial U^{\mu}}{\partial x}\bigg)^2  
			+ \bigg(\frac{\partial U^{\mu}}{\partial y}\bigg)^2 = 0,
		\end{equation*}
	since each term vanishes individually on $C$. This implies that $u := U^\mu(x,y,0)$ (considered as an
	analytic function defined in $\CC^2 \setminus \supp\, \mu$) satisfies one of the two equations
		\begin{equation*}
			\frac{\partial u}{\partial x} + i \frac{\partial u}{\partial y} = 0,\hspace{3.5mm} \text{or},
			\hspace{3.5mm}
			\frac{\partial u}{\partial x} - i \frac{\partial u}{\partial y} = 0.
		\end{equation*}
	By analytic continuation, the same equation holds on $\Gamma_\mu$. In other words, on 
	$\Gamma_\mu := \{u = const\}$, we have $\frac{\partial u}{\partial x} / \frac{\partial u}{\partial y} = \pm i$.
	Therefore, $\Gamma_\mu$ must be a complex line (incidentally, a  characteristic line for the two-dimensional 
	Laplacian), i.e., $\Gamma_\mu = \{(x,y) \in \CC^2 : x + iy = const., \text{ or } x - iy = const.\}$. In either
	case, the intersection $\Gamma_\mu \cap \RR^2$ is a point, not a curve $C$. This gives the desired contradiction.
	\end{proof}
	
	\subsubsection{Further observations}
	
	\begin{itemize}
	
	\item We remark that $\mu$ need not consist only of point charges. The argument extends to arbitrary charge 
	distributions $\mu$ as long as the curve $C$ doesn't `cut' through the support of $\mu$

	\item Extending this line of argument to $\RR^3$, assume that $C \subset \RR^3$ is a 1-dimensional degenerate 
	curve where $\nabla U^\mu = 0$, where $\mu$ is a point charge distribution in $\RR^3$. Then 
	$C \subset \Gamma_\mu \subset \CC^3$, where $\Gamma_\mu$ is an analytic hypersurface. It is clear that we 
	cannot claim that $\nabla U^{\mu} = 0$ on $\Gamma_\mu$, but we do have that
	$(\nabla U^{\mu})^2 = \big(\frac{\p U^\mu}{\p x}\big)^2 + \big(\frac{\p U^\mu}{\p y}\big)^2 + 
	\big(\frac{\p U^\mu}{\p z}\big)^2 = 0$ on $\Gamma_\mu$, i.e., $\Gamma_\mu$ is characteristic with respect
	to the Laplacian. Expanding the equation for $\Gamma_\mu$ in Taylor series, we find that the lowest
	nonzero homogeneous terms $v$ in the expansion still satisfy the eikonal equation $(\nabla v)^2 = 0$. As
	is shown on p. 178 of \cite{KL}, due to the homogeneity of $v$, we see that either $v$ is linear
	(i.e., $v(x,y,z) = \alpha x + \beta y + \gamma z$, with $\alpha^2 + \beta^2 + \gamma^2 = 0$), or the level set 
	$\{ v = const.\}$ must be the isotropic cone $\Gamma_0 := \{(x,y,z) \in \CC^3 : (x-x_0)^2 + (y-y_0)^2 + (z-z_0)^2
	 = 0\}$. The intersection of $\{ v = const.\}$ with $\RR^3$ then must be either a line, or a circle. Thus, up to 
	terms of order $3$ or higher, the curve $C$ must be either a circle or a line. Note that all known examples of
	degeneracy support this statement; see, for example, A.I. Janu\v{s}auskas' examples of a degenerate line through 
	the center of a square with alternating charges $\pm 1$ at the vertices, or the circle of radius 1, centered at the 
	origin, contained in the $y-z$ plane, perpendicular to the $x$-axis with charges $q$ at $x = \pm 1$, and 
	$-q/\sqrt{2}$ at the origin- cf. \cite{Y}.
		
	\item In general, the geometry of an equipotential surface in a neighborhood of a degenerate point where 
	the gradient vanishes is rather mysterious in dimensions higher than 2. Maxwell, for one, conjectured that,  
	similarly to the plane, the equipotential surface splits,  in the neighborhood of a critical (degenerate) point, 
	into several equidistributed hypersurfaces intersecting each other at equal angles. This is well known to be false 
	-- cf. \cite{Kellogg}, the footnote 1 on p. 276, end of Ch. X. Moreover, A. Szulkin, and later J.C. Wood have 
	constructed harmonic polynomials in $\mathbb{R}^3$ whose level set near a critical point (where the gradient 
	vanishes) is homeomorphic to a plane -- a shocking surprise, cf. the discussion in Sect. 2.5 of \cite{Duren} 
	(also, see \cite{Szulkin} and \cite{Wood}). 
	\end{itemize}

\section{Faraday's Problem.}
	We conclude this exposition wtih the following question, which came up in connection with the seemingly 
	unrelated problem  of uniqueness of the best uniform approximation by harmonic functions from approximation 
	theory \cite{KS}. However, the problem is, in spirit, very close to the subject of this paper. 
	Let $B = \{ x \in \RR^d : |x| < 1\}$ be the unit ball in $\RR^d$, $d > 2$, and $\mu$ be a (signed) charge 
	distribution supported on the closure of $B$ which produces the same electrostatic potential outside of 
	$\ol{B}$ as the point charge $\delta_0$ centered at the origin. In other words,
		\begin{equation}
			U^{\mu}(x) =  \int_{\ol{B}} \frac{d\mu(y)}{|x-y|^{d-2}} = \frac{1}{|x|}, \quad |x| > 1.
		\end{equation}
	In essence, this condition says that the effect of the \textit{signed} charge density $\mu$ is the same
	outside the ball as that of a \textit{positive} point charge placed at the origin. 
	\begin{conjecture} \label{Faraday-Problem}
		There is a positive charge distribution $m$, absolutely continuous with respect to $|\mu|$, 
		which produces the same potential $U^\mu$ as  does $\mu$ outside $B$.
	\end{conjecture}
	The requirement of absolute continuity here implies, in particular, that $m$ cannot contain any charge
	outside of the support of $\mu$.
	On physical grounds, this conjecture says that, since $U^\mu$ is the same as the potential of a 
	\textit{postive} charge $\delta_0$, one can expect that it is possible to `clean out' the support of $\mu$, 
	getting rid of all negative charges and redistributing the positive charge in such a way that the new charge
	produces the same effect outside of its support. Conjecture~\ref{Faraday-Problem} is true in dimension
	2, as shown in \cite{KS}. Yet, the techniques applied there relied heavily on analytic functions and are not 
	available in higher dimensions. However, the result seems reasonable (on physical grounds, at the very least)
	in all dimensions. Moreover, if this conjecture holds, then (as explained in \cite{KS}) it has deep and 
	important consequences for the difficult problems of uniqueness of best harmonic approximation in the uniform
	norm in dimensions larger than 2.


\end{document}